%% file: main.tex
\documentclass[11pt,a4paper]{article}

\usepackage{preprint}
\usepackage{mathtools}
\usepackage{xurl}

\title{Quantum Query Algorithms for the Constructive Diagonal Ramsey Theorem}
\author{Cheng Xin\\[0.4em]\small Department of Computer Science\\\small California State University, Fresno\\\small\texttt{cxin@mail.fresnostate.edu}}
\date{}
\AtBeginDocument{\hypersetup{
  pdftitle={Quantum Query Algorithms for the Constructive Diagonal Ramsey Theorem},
  pdfauthor={Cheng Xin},
  pdfsubject={Preprint: quantum query algorithms for constructive Ramsey search},
  pdfkeywords={diagonal Ramsey theorem, clique or independent set, quantum query complexity,
  adjacency oracle, implicit sets, total search}}}

\begin{document}

\maketitle

\begin{abstract}
The constructive diagonal Ramsey problem asks, given adjacency-oracle access
to an $N$-vertex graph, for a clique or independent set of the order
guaranteed by Ramsey's theorem.  We give a bounded-error quantum algorithm
that, for every $K\ge2$ and $N\ge4^{K-1}$, finds and verifies a homogeneous
$K$-set using
\[
  O\!\left(2^K K\log\frac K\eta\right)
\]
edge queries with failure probability at most $\eta$.  At the Ramsey scale
$N=2^n$, this yields a homogeneous set of order $\lfloor n/2\rfloor+1$
using $O(\sqrt N\log N\log(\log N/\eta))$ queries, improving on the $O(N)$
queries of the explicit classical recursion and giving, to our knowledge,
the first sublinear worst-case algorithm for the Ramsey relation.
We also derive an $\Omega(N^{1/12})$ quantum lower bound by a reduction
from collision finding.

The algorithm runs the constructive recursion over implicit candidate sets.
Each set is represented by a short conjunction of adjacency constraints and
sampled using capped unknown-solution quantum search, and a scale-aware
concentration schedule balances estimation accuracy against the increasing
cost of sampling deeper sets.  We complement the upper bound with an
$\Omega(N^{1-1/\sqrt2})$ randomized lower bound, transported from the
random-Painter analysis of online Ramsey numbers, which holds on the uniform
distribution $G(N,1/2)$.  On that distribution a greedy quantum search uses
only $\widetilde O(N^{1/4})$ queries, giving a provable polynomial quantum
speedup for Ramsey search on random graphs.  We also give an
estimation-free size-biased recursion and extend it to every fixed number
of edge colours.
\end{abstract}

\input{sections/01_introduction}
\input{sections/02_model_and_search}
\input{sections/03_algorithm}
\input{sections/04_lower_and_random}
\input{sections/05_related_work}
\input{sections/06_discussion}

\clearpage
\bibliographystyle{plainurl}
\bibliography{references,foundational_references}

\clearpage
\appendix
\input{appendices/A_capped_search}
\input{appendices/B_upper_proofs}
\input{appendices/C_lower_and_random}
\input{appendices/D_extensions_and_reproducibility}

\end{document}

%% file: sections/01_introduction.tex
\section{Introduction}\label{sec:introduction}

Ramsey's theorem guarantees that every graph on $N\ge4^{K-1}$ vertices
contains a clique or an independent set of order $K$.  Its classical
constructive proof \cite{ErdosSzekeres1935} repeatedly chooses a pivot and
restricts the candidate set to one of the pivot's two colour
neighbourhoods; truncating each retained neighbourhood to half its current
size, it finds a homogeneous $K$-set with $O(N)$ adjacency queries.  The
corresponding total search problem, \textsc{Ramsey}, has been studied in
proof complexity and in TFNP
\cite{Krajicek2001,Krajicek2005,KomargodskiNaorYogev2019,
PasarkarPapadimitriouYannakakis2023,JainLiRobereXun2024}: the input graph on
$N=2^n$ vertices is given by a circuit or an oracle, and the target order is
$K=n/2$.  The query complexity of this relation in the deterministic,
randomized, and quantum models is posed as an open question
in \cite[Section~III]{JainLiRobereXun2024}.  This paper answers the
quantum question up to a polynomial gap and clarifies the lower-bound
record on both the quantum and the classical side.

The central observation is that the candidate sets of the constructive proof
need not be materialized.  After $i$ rounds, membership in the current set
is determined by the $i$ adjacency constraints imposed by the previous
pivots.  These short predicates allow quantum search to sample from the
same recursion at substantially smaller query cost.  The threshold $4^{K-1}$
is the natural scale for this recursive construction.

\subsection{Our results}

We work in the promised valid-graph adjacency-oracle model.  One coherent
query returns the colour of a requested unordered vertex pair, and the output
is a set of vertices inducing one colour.  Our main result is the following.

\begin{table}[t]
\centering
\setlength{\tabcolsep}{4pt}
\begin{tabular}{llll}
\hline
model & input & lower bound & upper bound\\
\hline
randomized & worst case
 & $\Omega(N^{1-1/\sqrt2})$ \ [\cref{prop:classical-lower}]
 & $O(N)$ \ \cite{ErdosSzekeres1935}\\
randomized & $G(N,1/2)$
 & $\Omega(N^{1-1/\sqrt2})$ \ [\cref{prop:classical-lower}]
 & $\widetilde O(N^{1/2})$ \ [\cref{rem:classical-greedy}]\\
quantum & worst case
 & $\Omega(N^{1/12})$ \ [\cref{prop:quantum-lower}]
 & $\widetilde O(N^{1/2})$ \ [\cref{cor:succinct}]\\
quantum & $G(N,1/2)$
 & open
 & $\widetilde O(N^{1/4})$ \ [\cref{cor:separation}]\\
\hline
\end{tabular}
\caption{Bounded-error query complexity of finding a homogeneous set of
order $\lfloor n/2\rfloor+1$ in a graph on $N=2^n$ vertices, in the worst
case and on the uniform distribution $G(N,1/2)$; here
$1-1/\sqrt2\approx0.2929$.  The lower bounds are the random-Painter
estimate of \cite{ConlonFoxGrinshpunHe2019} and the collision reduction of
\cite{JainLiRobereXun2024} at multiplicity two (\cref{prop:quantum-lower}
is stated for order $\lfloor n/2\rfloor$ and hence valid here).}
\label{tab:bounds}
\end{table}

\begin{theorem}[Implicit-majority quantum Ramsey search]
\label{thm:main}
Let $K\ge2$, $N\ge4^{K-1}$, and $0<\eta<1/2$.  Given coherent edge-oracle
access to a simple undirected graph on $N$ vertices, there is a quantum
algorithm which, with probability at least $1-\eta$, outputs a $K$-clique or
a $K$-vertex independent set using
\[
 O\!\left(2^K K\log\frac K\eta\right)
\]
edge queries in the worst case.  The algorithm verifies every returned
witness and may output $\bot$ on the exceptional event.
\end{theorem}

At the succinct Ramsey scale, the theorem gives the following form.

\begin{corollary}\label{cor:succinct}
Let $n\ge2$, $N=2^n$, and $0<\eta<1/2$.  Given coherent edge-oracle access
to a simple graph on $N$ vertices, one can find a clique or independent set
of order $\lfloor n/2\rfloor+1$ using
\[
 O\!\left(
 \sqrt N\log N\log\frac{\log N}{\eta}
 \right)
\]
edge queries in the worst case.
\end{corollary}

\begin{proof}
Set $K=\lfloor n/2\rfloor+1$.  Then
$4^{K-1}=2^{2\lfloor n/2\rfloor}\le 2^n=N$, so \cref{thm:main} applies, and
$2^K\le2\sqrt N$, $K=O(\log N)$, and
$\log(K/\eta)=O(\log(\log N/\eta))$.
\end{proof}

To our knowledge, this is the first sublinear worst-case query algorithm
for the Ramsey relation in any model.  \Cref{tab:bounds} places it among
the known bounds together with the further results of this paper.

\paragraph{An estimation-free recursion.}
Following the colour of a uniformly sampled edge selects a colour class with
probability proportional to its size, and a one-dimensional survival
recurrence shows that this recursion completes with probability greater
than $1/2$ from $4^{K-1}$ vertices.  This gives an
$O(2^KK^2\log K\log(1/\eta))$-query algorithm without any majority
estimation (\cref{prop:size-biased-algorithm}), which extends to every
fixed number of edge colours (\cref{cor:multicolour}).

\paragraph{A randomized lower bound.}
Every randomized algorithm needs $\Omega(M^{1-1/\sqrt2})$ queries to find a
homogeneous $K$-set in a graph on $M=4^{K-1}$ vertices, even on the uniform
distribution $G(M,1/2)$ (\cref{prop:classical-lower}); this improves the
$M^{1/4}$ black-box bound of
\cite{ImpagliazzoNaor1988,KomargodskiNaorYogev2019}.  The proof transports
the random-Painter weight estimate of \cite{ConlonFoxGrinshpunHe2019} to
the query model, with a self-contained version in \cref{app:classical-lower}.

\paragraph{A quantum speedup on random graphs.}
On $G(N,1/2)$, a greedy quantum neighbourhood search finds a $K$-clique
with $\widetilde O(N^{1/4})$ queries (\cref{prop:random-graph}).  Together
with the distributional form of the lower bound, this gives a polynomial
separation between the quantum and randomized query complexities of the
Ramsey relation on its natural hard distribution (\cref{cor:separation}).
We are not aware of an earlier provable quantum speedup for this relation.
In the worst case, no separation is claimed: the randomized complexity is
only known to lie between $N^{1-1/\sqrt2}$ and $N$.

\paragraph{The quantum lower-bound record.}
The paragraph of \cite{JainLiRobereXun2024} that raises the query question
also reports an $N^{1-o(1)}$ quantum lower bound at the default parameters,
obtained by composing its collision reduction with the multicollision lower
bound of \cite{LiuZhandry2019}.  \Cref{thm:main} shows that no such bound
can hold at these parameters, and \cref{sec:jlrx} shows that the cited
ingredients compose to an exponent of at most $1/12$ (at most $1/24$ under
the parameter condition as printed), attained at multiplicity two.  The
resulting $\Omega(N^{1/12})$ bound (\cref{prop:quantum-lower}) is, to our
knowledge, the best known quantum lower bound; the main theorems of
\cite{JainLiRobereXun2024} do not depend on the remark.

\subsection{Techniques}

Both algorithms represent the candidate set after round $i$ by the
conjunction of the adjacency constraints generated by the first $i$ pivots,
so a membership test costs $O(i)$ edge queries, and sample it with capped
unknown-solution search (\cref{lem:capped-search}); conditional on success,
permutation symmetry makes each sample exactly uniform.  The scale-aware
recursion estimates a near-majority colour at every pivot.  Sampling a deep
set is exponentially more expensive than sampling an early one, so we
distribute a fixed error budget across the levels, doubling the allowed
error every six levels; a product invariant then keeps every candidate set
within a constant factor of its ideal size, and the query sum is dominated
by the last levels.

The lower bound and the random-graph speedup both rest on the predictable
size of common neighbourhoods in $G(N,1/2)$.  A martingale weight argument
in the style of \cite{ConlonFoxGrinshpunHe2019} shows that an adaptive
algorithm with $T$ queries exposes a monochromatic $K$-set with probability
at most $2\cdot2^{-\binom K2+c(c-1)}(2T)^{K-c}$ for every $c\le K/2$, and
optimizing $c$ gives the exponent $2-\sqrt2$; a union bound shows that
every $j$-set has common neighbourhood of density about $2^{-j}$, so greedy
quantum search costs $O(2^{j/2})$ predicate evaluations at depth $j$.

\paragraph{Organization.}
\Cref{sec:model} fixes the query model and the sampler.
\Cref{sec:algorithm} presents the two recursions, their analysis, and the
proof sketch of \cref{thm:main}; \cref{sec:lower} states the lower bounds
and the random-graph speedup.  \Cref{sec:prior-art} discusses related work
and the reported lower bound, and \cref{sec:scope} lists open questions.
All full proofs are in the appendices: the sampler and the encoding
corollary in \cref{app:capped-search}, the upper bounds in
\cref{app:upper-proofs}, the lower bounds and random-graph results in
\cref{app:classical-lower}, and the multicolour extension, numerical
illustrations, and AI disclosure in \cref{app:extensions}.

%% file: sections/02_model_and_search.tex
\section{Oracle model and capped quantum sampling}\label{sec:model}

Let $G\colon\binom{[N]}2\to\{0,1\}$ be a promised simple undirected graph,
extended symmetrically to ordered pairs with $G(u,u)=0$.  Its coherent edge
oracle acts as
$O_G\lvert u,v,b\rangle=\lvert u,v,b\oplus G(u,v)\rangle$.
Query complexity counts calls to $O_G$; reversible operations on vertex
labels are free, as usual in the query model.  The same asymptotic bound
applies to the standard locally verifiable TFNP relation with arbitrary
Boolean-circuit encodings.

\begin{corollary}[Locally verifiable arbitrary encoding]
\label{cor:unpromised}
The standard Ramsey TFNP relation that accepts either a local invalidity
certificate or a locally verified homogeneous set has the same
$O(2^K K\log(K/\eta))$ quantum-query upper bound as \cref{thm:main}.
\end{corollary}

The proof (\cref{app:capped-search}) canonicalizes the circuit into a
simple graph, runs \cref{thm:main}, and checks the $O(K^2)$ entries induced
by the output for a self-loop or an asymmetric pair.  Our algorithms use
the following consequence of unknown-solution quantum search
\cite{BBHT1998}, proved in \cref{app:capped-search}; the analysis needs
both bounded termination and exact conditional uniformity.

\begin{lemma}[Capped uniform marked-item search]\label{lem:capped-search}
Let $T\subseteq[M]$ have density at least $\lambda>0$, and suppose its
membership predicate costs $q$ edge queries.  There is a quantum block which
uses at most $O(q/\sqrt\lambda)$ edge queries, succeeds with probability at
least an absolute constant $p_0>0$, and, conditional on verified success,
returns an exactly uniform element of $T$.  The same predetermined cap makes
the block terminate if the density promise is false or $T$ is empty.
Moreover, $m$ independent uniform samples can be collected, except with
probability $\delta$, using a predetermined
$O(m+\log(1/\delta))$ blocks whenever the density promise holds.
\end{lemma}

For comparison, the explicit classical recursion on $M=4^{K-1}=2^{r+1}$
vertices, $r=2K-3$, queries at level $i$ all edges from a pivot to the other
$s_i=M/2^i$ candidates and retains exactly $s_i/2$ vertices of a majority
colour, for $\sum_{i<r}(s_i-1)<2M$ queries in total.  The quantum algorithms
retain the entire selected colour class but access it only implicitly.

%% file: sections/03_algorithm.tex
\section{Implicit-set recursions and the main theorem}\label{sec:algorithm}

\subsection{An estimation-free size-biased recursion}

Set $M=4^{K-1}$ and $r=2K-3$, restrict the input to its first $M$ vertices,
let $S_0=[M]$, and fix any $v_0\in S_0$.  At step
$i=0,\ldots,r-1$, define $T_i=S_i\setminus\{v_i\}$.  Apply the repeated
capped search from \cref{lem:capped-search}, with the universal density lower
bound $1/M$, to sample an exactly uniform $x_i\in T_i$; abort if the fixed
cap is exhausted.  Query $c_i=G(v_i,x_i)$ and set
\[
 S_{i+1}=\{x\in T_i:G(v_i,x)=c_i\},
 \qquad v_{i+1}=x_i.
\]
After the final step, put $w=v_r$.  Among the $r$ labels $c_i$, choose a
colour that appears at least $K-1$ times, and return the corresponding
pivots together with $w$ after directly verifying the induced edges.  The
algorithm never materializes $S_i$: its membership predicate is the
conjunction of the edge constraints imposed by the earlier pivots, together
with the exclusions of those pivots, and costs $O(i+1)$ edge queries.
Repeating each capped block $O(\log(2r))$ times reduces its conditional
failure probability on a nonempty set to $O(1/r)$.

For integers $d,s\ge0$, let $P_d(s)$ be the infimum, over all graphs,
candidate sets of size $s$, and distinguished pivots in those sets, of the
probability that the ideal size-biased recursion completes $d$ additional
steps; set $P_0(s)=1$ for $s\ge1$, $P_d(0)=0$, and $P_d(1)=0$ for $d\ge1$.

\begin{lemma}[Size-biased survival]\label{lem:size-biased-survival}
For every $d\ge0$ and $s\ge1$,
$P_d(s)\ge(s-2^d+1)_+/s$, where $(y)_+=\max\{y,0\}$.
In particular, for $M=4^{K-1}$ and $r=2K-3$, the ideal recursion completes
with probability at least $1/2+1/M$.
\end{lemma}

The proof is an induction on $d$: the sample selects the two colour classes
of sizes $a+b=s-1$ with probabilities $a/(s-1)$ and $b/(s-1)$, and
$(x)_++(y)_+\ge(x+y)_+$ collapses the resulting sum to one positive part.
Exact conditional uniformity couples an implemented run to the ideal
recursion, so charging each search failure probability $1/(100r)$ leaves
completion probability at least $0.49$; counting $O(\sqrt M)$ tests per
block, $O(\log(2r))$ blocks per level, and $O(i+1)$ queries per test then
gives the following bound (\cref{app:upper-proofs}).

\begin{proposition}[Size-biased quantum recursion]
\label{prop:size-biased-algorithm}
Under the hypotheses of \cref{thm:main}, the estimation-free algorithm returns
a verified homogeneous $K$-set with probability at least $1-\eta$ using
$O(K^2 2^K\log K\log(1/\eta))$ edge queries in the worst case.
\end{proposition}

\subsection{Scale-aware approximate majorities}

The sharper recursion estimates a near-majority colour at each pivot.  Fix
$M=4^{K-1}$ input vertices, $r=2K-3$, and $E=1/16$.  For $0\le i<r$, define
\[
 d_i=\left\lceil\frac{r-1-i}{6}\right\rceil,
 \qquad z_i=2^{-d_i},
 \qquad Z=\sum_{j=0}^{r-1}z_j,
\]
and allocate
\begin{equation}\label{eq:error-schedule}
 \varepsilon_i=E\frac{z_i}{Z},
 \qquad a_i=\frac12-\varepsilon_i,
\end{equation}
so that $\sum_i\varepsilon_i=E$.  The allowable error doubles once every six
levels, matching the increasing cost of sampling the deeper candidate sets.
Before round $i$, the recorded pivots $v_0,\ldots,v_{i-1}$ and colours
$c_0,\ldots,c_{i-1}$ define the exact implicit set
\begin{equation}\label{eq:implicit-set}
 S_i=\left\{
 u\in[M]\setminus\{v_0,\ldots,v_{i-1}\}:
 G(v_j,u)=c_j\ \text{for every }j<i
 \right\},
\end{equation}
whose reversible membership predicate uses $O(i)$ edge queries.

At round $i$, use \cref{lem:capped-search} to obtain a verified pivot
$v_i\in S_i$.  Then collect, with replacement,
\begin{equation}\label{eq:sample-count}
 m_i=\left\lceil
 \frac{1}{2\varepsilon_i^2}\log\frac{4r}{\eta}
 \right\rceil
\end{equation}
independent uniform samples from $S_i\setminus\{v_i\}$, query their edges to
$v_i$, and let $\widehat p_i$ be the observed fraction of colour~$1$.  Set
$c_i=1$ when $\widehat p_i\ge1/2$ and $c_i=0$ otherwise; appending this edge
constraint to \cref{eq:implicit-set} defines $S_{i+1}$.  After $r$ rounds,
one final capped search returns $w\in S_r$.  Select a colour that occurs at
least $K-1$ times among $c_0,\ldots,c_{r-1}$, take any $K-1$ pivots
carrying that colour, add $w$, query all $\binom K2$ induced edges, and
return the set only if this check confirms homogeneity.  Every search and
batch has a predetermined cap, so every execution terminates; a failed
search or batch returns $\bot$.

\subsection{Analysis}

\begin{lemma}[Conditional near-majorities]\label{lem:concentration}
Except with probability at most $\eta/2$, every completed sampling batch
reached after successful searches and an accurate prefix satisfies
$|\widehat p_i-p_i|\le\varepsilon_i$, where $p_i$ is the true colour-$1$
fraction from $v_i$ to $S_i\setminus\{v_i\}$.
\end{lemma}

Conditional on the classical history before level $i$ and on batch
success, the samples are independent and uniform on a fixed set, so
Hoeffding's inequality and \cref{eq:sample-count} bound the inaccuracy
probability by $\eta/(2r)$; a union bound over the first inaccurate level
gives the lemma.  Giving each of the $2r+1$ searches and batches failure
probability $\eta/(4r+2)$ adds at most $\eta/2$, so all searches and
estimates are good with probability at least $1-\eta$.  Let $s_i=|S_i|$,
$A_0=1$, and $A_i=\prod_{j<i}a_j$.

\begin{lemma}[Variable recurrence and density]\label{lem:survival}
On every prefix satisfying \cref{lem:concentration},
$s_i>MA_i-1$ and $A_i\ge\frac78\,2^{-i}$.  Consequently $S_r$ is nonempty,
$|S_i\setminus\{v_i\}|/M>\frac38\,2^{-i}$ for every sampling level $i<r$,
and $|S_r|/M\ge2^{-(r+1)}\ge\frac38\,2^{-r}$.
\end{lemma}

An accurate majority retains at least a fraction $a_i$ of the non-pivot
vertices, so $s_{i+1}\ge a_i(s_i-1)$; the bound on $A_i$ follows from
$\sum_i\varepsilon_i=1/16$, the invariant propagates because $a_i<1/2$,
and $M2^{-r}=2$ makes the integer $s_r$ exceed $3/4$
(\cref{app:upper-proofs}).

\begin{lemma}[Exact homogeneous output]\label{lem:output}
Whenever the algorithm reaches its final verified output, that output is a
$K$-clique or a $K$-vertex independent set.
\end{lemma}

\begin{proof}
For $i<j$, the definition of $S_{i+1}$ gives $G(v_i,v_j)=c_i$ and
$G(v_i,w)=c_i$; a colour occurring $\lceil r/2\rceil=K-1$ times among the
$r=2K-3$ colours therefore yields, with $w$, a homogeneous $K$-set, and the
final direct check is independent of earlier errors.
\end{proof}

\begin{proof}[Proof sketch of \cref{thm:main}]
With probability at least $1-\eta$ all searches succeed and all estimates
are accurate; then \cref{lem:survival} guarantees every density promise
used by the sampler and the nonemptiness of $S_r$, and \cref{lem:output}
gives a homogeneous output.  A capped block at level $i$ uses $O(2^{i/2})$
membership tests of $O(i+1)$ queries each, a batch uses $O(m_i)$ blocks,
and $Z<12$ gives $\varepsilon_i^{-2}=O(4^{d_i})$ in
\cref{eq:error-schedule}.  With $h=r-1-i$ and
$4^{\lceil h/6\rceil}\le4\cdot2^{h/3}$, the batches cost
\[
 O\!\left(\log\frac r\eta\sum_{i<r}(i+1)2^{i/2}4^{d_i}\right)
 =O\!\left(2^{r/2}\log\frac r\eta\sum_{h<r}(r-h)2^{-h/6}\right)
 =O\!\left(r2^{r/2}\log\frac r\eta\right),
\]
the amplified searches cost the same order, and $r=2K-3$ gives
$O(2^KK\log(K/\eta))$.  The full computation is in \cref{app:upper-proofs}.
\end{proof}

%% file: sections/04_lower_and_random.tex
\section{Lower bounds and a quantum speedup on random graphs}\label{sec:lower}

\subsection{A randomized classical lower bound}
\label{sec:classical-lower}

The following bound is the random-Painter estimate of Conlon, Fox,
Grinshpun, and He \cite[proof of Theorem~1]{ConlonFoxGrinshpunHe2019} for
online Ramsey numbers, transported to the query model.  Because that
estimate concerns uniformly random colourings, the bound holds on the
uniform distribution $G(N,1/2)$ and not only in the worst case.

\begin{proposition}[Randomized classical lower bound]
\label{prop:classical-lower}
Let $K\ge8$ and let $T=\lfloor2^{(2-\sqrt2)K-2}\rfloor$.  For every
$N\ge K$, every randomized algorithm that makes at most $T$ edge queries
outputs a homogeneous $K$-set of $G\sim G(N,1/2)$ with probability at most
$5/8$.  Consequently, for $M=4^{K-1}$, every randomized edge-query
algorithm that finds a homogeneous $K$-set with worst-case probability at
least $2/3$ uses
$\Omega(2^{(2-\sqrt2)K})=\Omega(M^{1-1/\sqrt2})$ queries.
\end{proposition}

The self-contained proof (\cref{app:classical-lower}) tracks, through the
matching number of the queried graph, the conditional probability that a
vertex set becomes a monochromatic clique: a $T$-query algorithm fully
exposes a monochromatic $K$-set with probability at most
$2\cdot2^{-\binom K2+c(c-1)}(2T)^{K-c}$ for every integer $c\le K/2$, which
is at most $1/4$ at $c=\lfloor(1-1/\sqrt2)K\rfloor$, and otherwise some pair
of the output is an unexposed fair coin, so the success probability is at
most $\frac14+\frac34\cdot\frac12=\frac58$.  The exponent
$1-1/\sqrt2\approx0.2929$ improves the $M^{1/4}$ black-box bound of
\cite{ImpagliazzoNaor1988,KomargodskiNaorYogev2019}; no worst-case
separation from \cref{thm:main} is claimed.

\subsection{A quantum speedup on random graphs}
\label{sec:random-graphs}

\begin{proposition}[Quantum greedy search on random graphs]
\label{prop:random-graph}
Let $K\ge2$, $M=4^{K-1}$, and $0<\eta<1/2$, and let $G\sim G(M,1/2)$.
There is a quantum algorithm which makes $O(K2^{K/2}\log(K/\eta))$ edge
queries on every input, never outputs a set that is not a $K$-clique, and
outputs a $K$-clique with probability at least $1-\eta-\pi_K$ over $G$ and
its internal randomness, where $\pi_K=2M^{K-1}e^{-2^{K-2}/12}$ satisfies
$\pi_K<e^{-100}$ for all $K\ge14$.
\end{proposition}

With probability at least $1-\pi_K$, every $j$-set with $j<K$ has at least
$M2^{-j-1}$ common neighbours, and the algorithm then grows a clique one
vertex at a time by the capped search of \cref{lem:capped-search} with
density parameter $2^{-j-1}$ at level $j$ (\cref{app:random-graph}): no
majority estimation is needed and the depth is $K$ rather than $2K-3$.

\begin{corollary}[Separation on random graphs at the Ramsey scale]
\label{cor:separation}
Let $n\ge26$, $N=2^n$, $K=\lfloor n/2\rfloor+1$, $0<\eta<1/2$, and
$G\sim G(N,1/2)$.
\begin{enumerate}
\item A quantum algorithm making
$O(N^{1/4}\log N\log(\log N/\eta))$ edge queries outputs a homogeneous
$K$-set of $G$ with probability at least $1-\eta-\pi_K$.
\item Every randomized algorithm making at most
$\lfloor2^{(2-\sqrt2)K-2}\rfloor=\Omega(N^{1-1/\sqrt2})$ edge queries
outputs a homogeneous $K$-set of $G$ with probability at most $5/8$.
\end{enumerate}
Since $1/4<1-1/\sqrt2$, the bounded-error quantum query complexity of the
Ramsey relation on $G(N,1/2)$ is polynomially smaller than its bounded-error
randomized query complexity.
\end{corollary}

The corollary applies \cref{prop:random-graph} to the first $4^{K-1}\le N$
vertices and \cref{prop:classical-lower} to $G(N,1/2)$
(\cref{app:random-graph}).

\begin{remark}[Classical greedy search]\label{rem:classical-greedy}
On the same expanding graphs, classical rejection sampling finds $v_{j+1}$
with $\lceil2^{j+1}\ln(K/\eta)\rceil$ trials of at most $j$ queries each,
except with probability $\eta/K$; so for $G\sim G(M,1/2)$ a classical
greedy search outputs a verified $K$-clique with probability at least
$1-\eta-\pi_K$ using $O(K2^K\log(K/\eta))=\widetilde O(N^{1/2})$ queries,
the classical entry of \cref{tab:bounds}.
\end{remark}

\subsection{A quantum lower bound from collision finding}
\label{sec:quantum-lower}

The reduction of Jain et al.\ \cite[Lemma~III.1]{JainLiRobereXun2024},
generalizing Komargodski, Naor, and Yogev
\cite{KomargodskiNaorYogev2019}, maps a function $h\colon[N]\to[H]$ and a
fixed graph $A_0$ on $[H]$ to the graph
\begin{equation}\label{eq:product-graph}
 A_h(u,v)=\mathbf 1[h(u)=h(v)]\lor A_0(h(u),h(v)),
 \qquad A_h(u,u)=0,
\end{equation}
on $[N]$.  Composed with the quantum collision lower bound, it gives the
best quantum lower bound we know; \cref{sec:jlrx} compares it with the
bound reported in \cite{JainLiRobereXun2024}.  A lower bound for target
order $\lfloor n/2\rfloor$ applies a fortiori to every larger target.

\begin{proposition}[Quantum lower bound from collision finding]
\label{prop:quantum-lower}
Let $N=2^n$ with $n\ge8$.  Every quantum algorithm that, given coherent
edge-oracle access to a simple graph on $N$ vertices, outputs a clique or
independent set of order $\lfloor n/2\rfloor$ with probability at least
$2/3$ on every input makes $\Omega(N^{1/12})$ edge queries.
\end{proposition}

In the proof (\cref{app:quantum-lower}), $H=2^{\lfloor n/4\rfloor}$ and
$A_0$ is an Erd\H{o}s graph on $[H]$ with no homogeneous set of order
$2\log_2H$.  Independent sets of $A_h$ have distinct $h$-values and map to
independent sets of $A_0$, cliques map to cliques, so every homogeneous set
of order at least $2\log_2H$ contains two vertices with the same
$h$-value; a $q$-query Ramsey algorithm therefore finds a collision of a
random $h$ with $O(q+n)$ queries, which requires
$\Omega(H^{1/3})=\Omega(N^{1/12})$ \cite{Zhandry2015}.

%% file: sections/05_related_work.tex
\section{Related work and the reported lower bound}
\label{sec:prior-art}

\paragraph{Quantum algorithms for Ramsey problems.}
Quantum approaches based on adiabatic computation, annealing, or quantum
counting search for colourings that establish or test Ramsey-number bounds
\cite{GaitanClark2012,BianEtAl2013,Wang2016,QuLiWangBaoCao2013,
RanjbarMacreadyClarkGaitan2016,PionMniszewski2025}; our input is one fixed
graph given by an oracle.  Generic quantum algorithms
for clique, independent set, and fixed-subgraph containment
\cite{Doern2005,ChildsEisenberg2005,LeeMagniezSantha2012,Zhu2012} and
quantum backtracking \cite{Montanaro2018} operate in a closer oracle model;
quantum walks also give sublinear-query algorithms for fixed
sub-hypergraph containment \cite{LeGallNishimuraTani2016}.  We are not aware
of an earlier implementation of the Erd\H{o}s--Szekeres
recursion over implicit sets in the adjacency-query model.

\paragraph{Constructive Ramsey search.}
Ramsey search is connected to weak and structured pigeonhole principles
\cite{Krajicek2001,Krajicek2005}, its white-box hardness follows from
collision-resistant hashing \cite{KomargodskiNaorYogev2019}, and it belongs
to the polynomial long-choice class PLC
\cite{PasarkarPapadimitriouYannakakis2023}.  The closest matching query
relation is studied in \cite{JainLiRobereXun2024}, which shows that it is
not black-box reducible to PIGEON and raises the question of its query
complexity in the deterministic, randomized, and quantum models.  The
online Ramsey game (see \cite{ConlonFoxGrinshpunHe2019}), in which Builder
queries edges and Painter answers adversarially, is the deterministic
version of this question with an unbounded vertex set; its best known lower
bound is the random-Painter bound used in \cref{prop:classical-lower}.

\paragraph{Comparison with a reported quantum lower bound.}
\label{sec:jlrx}
\emph{Model and parameter alignment.}
Definition~II.6 of \cite{JainLiRobereXun2024} uses $N=2^n$ vertices, and the
convention following Definition~II.7 sets the default target to $K=n/2$.
The unnumbered ``Query complexity of RAMSEY'' paragraph in Section~III
(p.~415) reports an $N^{1-o(1)}$ quantum lower bound, based on its
Lemma~III.1 and the multicollision lower bound of \cite{LiuZhandry2019}.
Lemma~III.1 produces the product graph \cref{eq:product-graph}, a simple
graph whose coherent queries cost $O(1)$ queries to $h$, so the reduction
lies within our promise; \cref{cor:unpromised} covers arbitrary circuits,
and \cref{cor:succinct} meets the default target after discarding one
vertex.  The two formulations therefore align.

\emph{Exponent obtained by direct substitution.}
Let $H$ denote the multicollision range size and $G$ the number of vertices
in the resulting Ramsey instance.  For fixed multiplicity $t$, the
sufficient condition in Theorem~I.3 of \cite{JainLiRobereXun2024} is
$G\ge H^{4t}/4^t$.  At the boundary $G=\Theta_t(H^{4t})$, the exponent
$\alpha_t=(2^{t-1}-1)/(2^t-1)$ of \cite{LiuZhandry2019} transfers to $G$ as
$\alpha_t/(4t)=(2^{t-1}-1)/(4t(2^t-1))\le1/24$ for $t\ge2$, with equality
at $t=2$; for $t\ge3$, $\alpha_t<1/2$ gives $\alpha_t/(4t)<1/(8t)\le1/24$.  Increasing $G$ relative to $H$ only reduces
the transferred exponent, and the multicollision theorem is stated for
fixed $t$, so the cited result does not support a growing-multiplicity
limit either.  The printed condition is not tight: the proof of
Lemma~III.1 only uses that a homogeneous set of the product graph maps to
at most $K_0-1$ values of $h$, where $K_0=2\log_2H$ is the forbidden order
in $A_0$, so a homogeneous set of order $(t-1)(K_0-1)+1$ already forces a
$t$-collision, the reduction only needs $G=\Theta_t(H^{4(t-1)})$, and the
transferred exponent improves to $\alpha_t/(4(t-1))$, which is $1/12$ at
$t=2$ and below $1/(8(t-1))\le1/16$ for $t\ge3$.  Direct composition of
the cited ingredients therefore yields an exponent at most $1/12$, rather
than $1-o(1)$; \cref{prop:quantum-lower} is exactly this $t=2$ bound, and
\cref{thm:main} rules out an $N^{1-o(1)}$ lower bound at these parameters.
The remark is not used elsewhere in \cite{JainLiRobereXun2024}, and the
black-box separation theorems of that paper are unaffected.

%% file: sections/06_discussion.tex
\section{Discussion and open questions}
\label{sec:scope}

The query improvement rests on storing each candidate set as the
conjunction of the adjacency constraints created by earlier pivots, and on
allocating estimation error according to the cost of each level.
\Cref{tab:bounds} leaves two central questions.

\paragraph{A worst-case quantum advantage.}
Can the quantum and randomized query complexities be separated in the worst
case?  The randomized complexity on $N=2^n$ vertices lies between
$\Omega(N^{1-1/\sqrt2})$ and $O(N)$, while the quantum upper bound is
$\widetilde O(\sqrt N)$.  A randomized lower bound of
$N^{1/2+\Omega(1)}$ would establish such a separation.  The implicit-set
recursion explains the difference between the available upper bounds:
classical rejection sampling at depth $i$ needs $2^i$ membership tests,
whereas quantum search needs $2^{i/2}$.

\paragraph{The optimal quantum exponent.}
Is $\sqrt N$ optimal in the worst case?  The current quantum bounds are
$\Omega(N^{1/12})$ and $O(\sqrt N\log N\log\log N)$.  The final levels of
the recursion search sets of constant size inside a universe of size
$\Theta(N)$, but these sets are defined by $O(\log N)$ adjacency
constraints, so the optimality of Grover search does not establish a lower
bound for this relation.

\paragraph{Scope.}
The distributional separation on $G(N,1/2)$ in \cref{cor:separation} and
the multicolour bound in \cref{cor:multicolour} hold under their stated
contracts.  \Cref{app:experiments} reports small sanity-check experiments;
the query bounds rest on the proofs in the appendices.  The paper claims
query bounds only, not gate complexity or hardware performance.

%% file: appendices/A_capped_search.tex
\section{Proofs of the sampler lemma and the encoding corollary}\label{app:capped-search}

\begin{proof}[Proof of \cref{lem:capped-search}]
Begin in the uniform superposition on $[M]$ and apply the randomized-iteration
unknown-solution search of \cite{BBHT1998}.  When the marked density is at
least $\lambda$, truncate the geometric iteration schedule once its maximum
iteration count reaches a sufficiently large constant multiple of
$1/\sqrt\lambda$.  The cited analysis gives a success probability $p_0>0$,
independent of $M$ and $\lambda$, using $O(1/\sqrt\lambda)$ calls to the
membership predicate.  Because the cap is fixed before any oracle outcome is
observed, the block terminates for every predicate.

For each fixed number of Grover iterations, all marked basis states have the
same amplitude.  Randomizing the iteration count preserves this permutation
symmetry.  After measuring a candidate, evaluate the exact membership
predicate and accept only a marked outcome.  Conditional on acceptance, the
output is therefore uniform on the marked set $T$.  Independent blocks
reinitialize all registers and random choices, so their accepted outputs are
independent conditional on the block-success indicators.

To collect a batch, run
\[
 B=C\bigl(m+\log(1/\delta)\bigr)
\]
independent capped blocks, where the absolute constant $C$ is chosen from
$p_0$.  The number of accepted blocks stochastically dominates
$\operatorname{Bin}(B,p_0)$.  A Chernoff lower-tail bound
\cite{Chernoff1952} makes the probability of fewer than $m$ acceptances at
most $\delta$.  Conditional on any success pattern with at least $m$
acceptances, the first $m$ accepted outputs are independent and uniform on
$T$.
\end{proof}

\begin{proof}[Proof of \cref{cor:unpromised}]
Given a circuit $C(u,v)$, define the canonical simple graph
\[
 A(u,v)=
 \begin{cases}
 C(\min\{u,v\},\max\{u,v\}),&u\ne v,\\
 0,&u=v.
 \end{cases}
\]
One query to $A$ uses $O(1)$ queries to $C$, including uncomputation.  Run
\cref{thm:main} on $A$ and let $W$ be the returned homogeneous set.  Query the
diagonal and both ordered entries induced by $W$, using $O(K^2)$ additional
queries.  A self-loop or an asymmetric pair is a local invalidity certificate.
Otherwise $C$ agrees with $A$ on $W$, so $W$ is a valid clique or
independent-set witness.  The verification cost is absorbed by the bound in
\cref{thm:main}.
\end{proof}

%% file: appendices/B_upper_proofs.tex
\section{Proofs for \texorpdfstring{\cref{sec:algorithm}}{Section 3}}
\label{app:upper-proofs}

\subsection{The size-biased recursion}

\begin{proof}[Proof of \cref{lem:size-biased-survival}]
The case $d=0$ is immediate.  Suppose $d\ge1$ and $s\ge2$, and let the two
colour classes among the $s-1$ non-pivot vertices have sizes $a$ and $b$.
A uniform sample selects these classes with probabilities $a/(s-1)$ and
$b/(s-1)$.  The induction hypothesis therefore gives
\[
 P_d(s)\ge
 \frac{aP_{d-1}(a)+bP_{d-1}(b)}{s-1}
 \ge
 \frac{(a-2^{d-1}+1)_+ +(b-2^{d-1}+1)_+}{s-1}.
\]
Because $(x)_++(y)_+\ge(x+y)_+$ and $a+b=s-1$, the numerator is at least
$(s-2^d+1)_+$.  Replacing the denominator $s-1$ by $s$ weakens the bound and
proves the claim; the definitions cover $s\le1$.  Finally, $M=2^{r+1}$, so
the bound at $(d,s)=(r,M)$ is
$(M-2^r+1)/M=1/2+1/M$.
\end{proof}

Exact conditional uniformity of each successful capped search couples an
implemented run to the ideal recursion.  Assign each of the $r$ searches
failure probability at most $1/(100r)$.  A union bound decreases the
completion probability by at most $1/100$, so a single fixed-cap run succeeds
with probability at least $0.49$.  Independent repetition, followed by exact
verification, amplifies this probability to $1-\eta$ using
$O(\log(1/\eta))$ runs.  On every completed run, nesting gives
$G(v_i,v_j)=G(v_i,w)=c_i$ for $j>i$; the pigeonhole argument in
\cref{lem:output} then gives a homogeneous output.

\begin{proof}[Proof of \cref{prop:size-biased-algorithm}]
A capped sample from a nonempty subset of the $M$-element universe costs
$O(\sqrt M)$ membership tests.  Repeating the block $O(\log(2r))$ times
reduces its failure probability to $O(1/r)$.  Since the membership predicate
at level $i$ costs $O(i+1)$ edge queries, one fixed-cap run uses
\[
 O\!\left(\sqrt M\log(2r)\sum_{i=0}^{r-1}(i+1)\right)
 =O(K^2 2^K\log K)
\]
edge queries.  The survival analysis above gives a constant success
probability for each run, including capped-search failures and final
verification.  Repeating independent runs $O(\log(1/\eta))$ times proves the
stated worst-case bound.
\end{proof}

\subsection{The scale-aware recursion}

\begin{proof}[Proof of \cref{lem:concentration}]
Condition on the complete classical history $\mathcal H_i$ before the samples
at level $i$.  The graph, $S_i$, and $v_i$ are then fixed.  Let
$\mathcal B_i$ be the event that the verified batch succeeds.  Conditional on
$\mathcal B_i$, \cref{lem:capped-search} gives independent uniform samples
from the fixed set.  Hoeffding's inequality \cite{Hoeffding1963} and
\cref{eq:sample-count} imply
\[
 \Pr\!\left(
 |\widehat p_i-p_i|>\varepsilon_i\mid\mathcal H_i,\mathcal B_i
 \right)
 \le2e^{-2m_i\varepsilon_i^2}
 \le\frac{\eta}{2r}.
\]
Multiplying by $\Pr(\mathcal B_i\mid\mathcal H_i)\le1$ gives the same upper
bound for a completed but inaccurate batch conditional on $\mathcal H_i$.
Apply this bound to the first inaccurate completed level after an accurate
prefix.  The tower property removes the conditioning, and a union bound over
the $r$ possible levels completes the proof.  The argument does not require
independence across levels.
\end{proof}

There are $r$ pivot searches, $r$ sampling batches, and one final search.
Give each operation failure probability at most $\eta/(4r+2)$ using the
capped amplification in \cref{lem:capped-search}.  Applying these conditional
bounds after good prefixes, as above, shows that their union has probability
at most $\eta/2$.  Together with \cref{lem:concentration}, all searches and
estimates used below are good with probability at least $1-\eta$.

\begin{proof}[Proof of \cref{lem:survival}]
An accurate empirical majority retains at least the fraction
$a_i=1/2-\varepsilon_i$ of the non-pivot vertices.  Hence
\begin{equation}\label{eq:variable-recurrence}
 s_{i+1}\ge a_i(s_i-1).
\end{equation}
Since $\sum_i\varepsilon_i=1/16$ and
$\prod_j(1-x_j)\ge1-\sum_jx_j$ for $x_j\in[0,1]$,
\[
 A_i
 =2^{-i}\prod_{j<i}(1-2\varepsilon_j)
 \ge2^{-i}\left(1-2\sum_{j<i}\varepsilon_j\right)
 \ge\frac78\,2^{-i}.
\]

We prove $s_i>MA_i-1$ by induction.  The claim holds at $i=0$.  If it holds
at level $i$, then \cref{eq:variable-recurrence} and $a_i<1/2$ give
\[
 s_{i+1}>a_i(MA_i-2)=MA_{i+1}-2a_i>MA_{i+1}-1.
\]
Because $M2^{-r}=2$, we obtain
$s_r>(7/8)2-1=3/4$, and therefore the integer $s_r$ is positive.  For $i<r$,
we have $M2^{-i}\ge4$ and
\[
 s_i-1>MA_i-2
 \ge\frac78M2^{-i}-2
 \ge\frac38M2^{-i}.
\]
Finally, $s_r\ge1$ and $1/M=2^{-(r+1)}$, which proves the last assertion.
\end{proof}

\begin{proof}[Proof of \cref{thm:main}]
Set $M=4^{K-1}$ and run the algorithm of \cref{sec:algorithm} on any fixed
set of $M$ input vertices.  The failure allocation above, together with
\cref{lem:concentration}, implies that all required searches succeed and all
completed estimates are accurate with probability at least $1-\eta$.  On
this event, \cref{lem:survival} guarantees every density promise used by the
sampler and ensures that $S_r$ is nonempty.  The final output is homogeneous
by \cref{lem:output}.  Because the algorithm verifies the returned set
directly, every non-$\bot$ output is valid even outside the good event.

It remains to bound the number of queries.  By \cref{lem:survival}, a capped
search block at level $i$ uses $O(2^{i/2})$ membership tests.  Computing and
uncomputing membership and querying the sampled edge costs $O(i+1)$ edge
queries.  Define
\[
 b_i=(i+1)2^{i/2}.
\]
The batch cap in \cref{lem:capped-search} uses
$O(m_i+\log(r/\eta))=O(m_i)$ blocks.  By \cref{eq:sample-count}, the total
cost of the sampling batches is therefore
\begin{equation}\label{eq:query-sum}
 O\!\left(
 \log\frac r\eta
 \sum_{i=0}^{r-1}b_i\varepsilon_i^{-2}
 \right).
\end{equation}

At most six indices share any value of $d_i$, and hence
\[
 Z=\sum_i2^{-d_i}<6\sum_{d\ge0}2^{-d}=12.
\]
Using \cref{eq:error-schedule},
\[
 \sum_{i=0}^{r-1}b_i\varepsilon_i^{-2}
 =O\!\left(
 \sum_{i=0}^{r-1}(i+1)2^{i/2}4^{d_i}
 \right).
\]
Write $h=r-1-i$.  Since
$4^{\lceil h/6\rceil}\le4\cdot2^{h/3}$, the last display is at most
\[
 O\!\left(
 2^{r/2}\sum_{h=0}^{r-1}(r-h)2^{-h/6}
 \right)
 =O(r2^{r/2}).
\]
The amplified pivot searches and the final search cost
\[
 O\!\left(\log\frac r\eta\sum_{i=0}^{r}b_i\right)
 =O\!\left(r2^{r/2}\log\frac r\eta\right),
\]
and the $O(K^2)$ output verification is smaller.  Substituting $r=2K-3$
into \cref{eq:query-sum} gives
\[
 O\!\left(2^K K\log\frac K\eta\right),
\]
as claimed.  Each search begins in the uniform state on the fixed
power-of-two universe $[M]$, and the recorded pivot labels and colours define
the reversible membership predicate for $S_i$, so the algorithm requires no
materialized representation of any candidate set.
\end{proof}

%% file: appendices/C_lower_and_random.tex
\section{Proofs for \texorpdfstring{\cref{sec:lower}}{Section 4}}
\label{app:classical-lower}

\subsection{A weight bound for adaptive queries to a random graph}
\label{app:weight-bound}

We record the random-Painter estimate of Conlon, Fox, Grinshpun, and He
\cite[proof of Theorem~1]{ConlonFoxGrinshpunHe2019} in the form used in
this paper.  We follow their weight argument, taking the matching number
of the queried graph as the potential; removing the critical edge from a
maximum matching then leaves a matching of the smaller set, which keeps the
recursion closed.

Throughout this subsection $V\ge2$, and $G\sim G(V,1/2)$ assigns
independent uniform colours in $\{0,1\}$ to the pairs of $[V]$.  A
deterministic algorithm queries $T$ distinct pairs adaptively.  Let
$e_1,\ldots,e_T$ be the queried pairs, let $Q_t$ be the graph formed by
$e_1,\ldots,e_t$, and let $\mathcal F_t$ be the $\sigma$-algebra generated
by the first $t$ answers.  Because the algorithm is deterministic, $e_{t+1}$
is an $\mathcal F_t$-measurable pair distinct from $e_1,\ldots,e_t$, and its
colour is therefore a uniform bit independent of $\mathcal F_t$.

For $U\subseteq[V]$ and $0\le t\le T$, let $e_t(U)$ be the number of pairs
inside $U$ among $e_1,\ldots,e_t$, and call $U$ \emph{consistent at time
$t$} if all of these pairs have colour~$1$.  Define
\[
 w(U,t)=
 \begin{cases}
 2^{-\binom{|U|}2+e_t(U)}, & U\text{ is consistent at time }t,\\
 0,&\text{otherwise},
 \end{cases}
\]
the conditional probability, given $\mathcal F_t$, that $U$ becomes a
colour-$1$ clique when its remaining pairs are revealed.  Let $\nu(U,t)$ be
the matching number of $Q_t[U]$; it is nondecreasing in $t$ and increases by
at most one when a pair is added.  For integers $k\ge2c\ge0$ set
\[
 w_{k,c}(t)=\sum_{|U|=k,\ \nu(U,t)\ge c}w(U,t).
\]

\begin{lemma}[Weight bound]\label{lem:weight-bound}
For all integers $k\ge2c\ge0$,
\[
 \mathbb E\,w_{k,c}(T)\le 2^{-\binom k2+c(c-1)}\,V^{k-2c}\,T^{c}.
\]
\end{lemma}

\begin{proof}
\emph{Martingale property.}  Fix $U$.  If $e_{t+1}\not\subseteq U$, then
$w(U,t+1)=w(U,t)$.  If $e_{t+1}\subseteq U$, then with probability $1/2$
the new pair has colour~$1$ and $w(U,t+1)=2w(U,t)$, and otherwise
$w(U,t+1)=0$.  Hence $\mathbb E[w(U,t+1)\mid\mathcal F_t]=w(U,t)$, and in
particular $\mathbb E\,w(U,T)=w(U,0)=2^{-\binom{|U|}2}$.  Summing over all
$k$-sets gives the case $c=0$:
\[
 \mathbb E\,w_{k,0}(T)=\binom Vk2^{-\binom k2}\le V^k2^{-\binom k2}.
\]

\emph{Monotonicity.}  Since $\nu(U,\cdot)$ is nondecreasing,
$w_{k,c}(t+1)\ge\sum_{\nu(U,t)\ge c}w(U,t+1)$, and each event
$\{\nu(U,t)\ge c\}$ is $\mathcal F_t$-measurable.  Taking conditional
expectations gives $\mathbb E\,w_{k,c}(t)\le\mathbb E\,w_{k,c}(t+1)$.

\emph{Critical times.}  Let $c\ge1$.  For each $U$ let
$\tau(U)=\min\{t:\nu(U,t)\ge c\}$, with $\tau(U)=\infty$ if
$\nu(U,T)<c$; this is a stopping time with $\tau(U)\ge1$.
Since $\{\tau(U)=t\}\in\mathcal F_t$ and $w(U,\cdot)$ is a martingale,
\[
 \mathbb E\,w_{k,c}(T)
 =\sum_{|U|=k}\sum_{t=1}^T\mathbb E\bigl[w(U,T)\,\mathbf 1[\tau(U)=t]\bigr]
 =\sum_{t=1}^T\mathbb E\,w^*_{k,c}(t),
\]
where
\[
 w^*_{k,c}(t)=\sum_{|U|=k,\ \tau(U)=t}w(U,t).
\]

\emph{Charging.}  Fix $t$ and write $e_t=\{x,y\}$.  Suppose $|U|=k$ and
$\tau(U)=t$.  Then $\nu(U,t-1)=c-1$, $\nu(U,t)=c$, and $e_t\subseteq U$.
Moreover, every maximum matching of $Q_t[U]$ contains $e_t$, since a
maximum matching avoiding $e_t$ would be a matching of size $c$ in
$Q_{t-1}[U]$.  Fix such a matching $M$, let $M'=M\setminus\{e_t\}$, let $A$
be the set of $2(c-1)$ vertices covered by $M'$, and let
$U'=U\setminus\{x,y\}$.  Then $M'$ is a matching of $Q_{t-1}[U']$, so
$\nu(U',t-1)\ge c-1$.  We claim that at most $2(c-1)$ of the pairs between
$\{x,y\}$ and $U'$ belong to $Q_{t-1}$.  Indeed, if $xa\in Q_{t-1}$ with
$a\in U'\setminus A$, then $M'\cup\{xa\}$ is a matching of size $c$ in
$Q_{t-1}[U]$, a contradiction; the same holds for $y$.  For each
$ab\in M'$, the pairs $xa$ and $yb$ cannot both lie in $Q_{t-1}$, since
$(M'\setminus\{ab\})\cup\{xa,yb\}$ would again be a matching of size $c$ in
$Q_{t-1}[U]$; likewise $xb$ and $ya$ cannot both lie in $Q_{t-1}$.  Hence
at most two of $xa,xb,ya,yb$ lie in $Q_{t-1}$, and the claim follows by
summing over the $c-1$ edges of $M'$.

Consequently $e_t(U)\le e_{t-1}(U')+1+2(c-1)$, and $U$ is consistent at
time $t$ only if $U'$ is consistent at time $t-1$ and $e_t$ has colour~$1$.
Using $\binom k2-\binom{k-2}2=2k-3$, we obtain
\[
 w(U,t)\le 2^{-(2k-2c-2)}\,w(U',t-1)
 \quad\text{if }e_t\text{ has colour }1,
\]
and $w(U,t)=0$ otherwise.  The map $U\mapsto U'$ is injective because
$U=U'\cup e_t$.  Therefore
\[
 w^*_{k,c}(t)\le
 \mathbf 1[e_t\text{ has colour }1]\;2^{-(2k-2c-2)}\,w_{k-2,c-1}(t-1).
\]
The indicator is independent of $\mathcal F_{t-1}$ with mean $1/2$, while
$w_{k-2,c-1}(t-1)$ is $\mathcal F_{t-1}$-measurable.  Taking expectations
and using monotonicity,
\[
 \mathbb E\,w^*_{k,c}(t)\le 2^{-(2k-2c-1)}\,\mathbb E\,w_{k-2,c-1}(T),
\]
and summing over $t\le T$ gives
\[
 \mathbb E\,w_{k,c}(T)\le T\,2^{-(2k-2c-1)}\,\mathbb E\,w_{k-2,c-1}(T).
\]

\emph{Iteration.}  Applying the last inequality $c$ times, the exponents
$2(k-2i)-2(c-i)-1$ for $i=0,\ldots,c-1$ sum to $2kc-3c^2$, and the case
$c=0$ gives
\[
 \mathbb E\,w_{k,c}(T)
 \le T^c\,2^{-(2kc-3c^2)}\,V^{k-2c}\,2^{-\binom{k-2c}2}
 =2^{-\binom k2+c(c-1)}\,V^{k-2c}\,T^c,
\]
because $2kc-3c^2+\binom{k-2c}2=\binom k2-c(c-1)$.
\end{proof}

\begin{corollary}[Fully exposed monochromatic sets]\label{cor:exposed}
Let $K\ge2$, let $V\ge K$, let $T\ge1$, and let a deterministic algorithm
make at most $T$ adaptive edge queries to $G\sim G(V,1/2)$.  For every integer
$0\le c\le K/2$, the expected number of $K$-sets all of whose $\binom K2$
pairs have been queried and carry the same colour is at most
\[
 2\cdot2^{-\binom K2+c(c-1)}(2T)^{K-c}.
\]
\end{corollary}

\begin{proof}
Simulate the algorithm while renaming vertices in order of first
appearance in a query, using a table from original to canonical labels.
This does not change the joint distribution of the queried graph and its
colours, because the colours of distinct pairs are independent and uniform;
the renamed algorithm queries pairs inside $[V']$ with $V'=\min\{V,2T\}$,
and the number of fully exposed monochromatic $K$-sets is unchanged.  After
the renamed algorithm halts, append distinct unqueried pairs inside $[V']$
in a fixed order until exactly $T'=\min\{T,\binom{V'}2\}$ pairs have been
queried; this can only increase that number.  A fully exposed colour-$1$
$K$-set $U$ satisfies $\nu(U,T')=\lfloor K/2\rfloor\ge c$ and $w(U,T')=1$,
so the number of such sets is at most $w_{K,c}(T')$.
\Cref{lem:weight-bound} with $V'\le2T$ and $T'\le T$ bounds its expectation
by
$2^{-\binom K2+c(c-1)}(2T)^{K-2c}T^c\le2^{-\binom K2+c(c-1)}(2T)^{K-c}$.
The same bound holds for colour~$0$ by symmetry.
\end{proof}

\subsection{Proof of the randomized lower bound}
\label{app:classical-lower-proof}

\begin{proof}[Proof of \cref{prop:classical-lower}]
Let $\beta=2-\sqrt2$ and $\alpha=1-1/\sqrt2$, so that $\beta=2\alpha$ and
$\alpha^2-2\alpha+1/2=0$.  Set
\[
 T=\bigl\lfloor2^{\beta K-2}\bigr\rfloor,
 \qquad
 c=\lfloor\alpha K\rfloor=\alpha K+\epsilon,
 \qquad -1<\epsilon\le0.
\]
Since $\alpha<1/2$, we have $0\le c\le K/2$, and $K\ge8$ gives $T\ge1$.

Consider first a deterministic algorithm making at most $T$ queries to
$G\sim G(N,1/2)$.  By \cref{cor:exposed} and Markov's inequality, the
probability that the queried pairs contain a fully exposed monochromatic
$K$-set is at most
\begin{equation}\label{eq:random-painter-bound}
 2\cdot
 2^{-\binom K2+c(c-1)}(2T)^{K-c}.
\end{equation}
Since $2T\le2^{\beta K-1}$, the base-$2$ logarithm of
\cref{eq:random-painter-bound} is at most
\[
 1-\binom K2+c(c-1)+(K-c)(\beta K-1)
 =1-\frac K2+\epsilon^2,
\]
where the identity follows by substituting $c=\alpha K+\epsilon$ and
$\beta=2\alpha$ and using $\alpha^2-2\alpha+1/2=0$.  For $K\ge8$ this
quantity is at most $-2$, so the event has probability at most $1/4$.

Now condition on the query transcript and suppose that the algorithm
outputs a $K$-set whose internal pairs have not all been exposed in one
colour.  If the exposed internal pairs contain both colours, the output is
not homogeneous.  If they contain exactly one colour, at least one
unexposed pair must match that colour, which has conditional probability at
most $1/2$, because unexposed pairs are independent fair coins given the
transcript.  If no internal pair has been exposed, the success probability
is $2^{1-\binom K2}\le1/2$ for $K\ge3$.  Thus the conditional success
probability is at most $1/2$ outside the fully exposed event, and the
success probability under $G(N,1/2)$ is at most
\[
 \frac14+\frac34\cdot\frac12=\frac58.
\]

A randomized algorithm making at most $T$ queries is a distribution over
such deterministic algorithms, so its success probability on
$G\sim G(N,1/2)$ is also at most $5/8$.  Finally, an algorithm whose
worst-case success probability on graphs with $M$ vertices is at least
$2/3>5/8$ must make more than $T$ queries, and $M=4^{K-1}$ gives
\[
 2^{\beta K}
 =\Theta\!\left(M^{\beta/2}\right)
 =\Theta\!\left(M^{1-1/\sqrt2}\right),
\]
which proves the proposition.
\end{proof}

\subsection{Proof of the random-graph speedup}
\label{app:random-graph}

\begin{proof}[Proof of \cref{prop:random-graph}]
Let $M=4^{K-1}$.  For $W\subseteq[M]$ let
\[
 C(W)=\{u\in[M]\setminus W:\ G(w,u)=1\text{ for every }w\in W\}
\]
be the common neighbourhood of $W$.  Call $G$ \emph{expanding} if
$|C(W)|\ge M2^{-j-1}$ for every $j\in\{1,\ldots,K-1\}$ and every $j$-set
$W$; for $W=\emptyset$ the bound $|C(\emptyset)|=M$ holds trivially.

\emph{Random graphs are expanding.}  Fix $1\le j\le K-1$ and a $j$-set
$W$.  Then $|C(W)|\sim\operatorname{Bin}(M-j,2^{-j})$ with mean
$\mu=(M-j)2^{-j}$.  Since $M=4^{K-1}\ge4(K-1)\ge4j$, we have
$\mu\ge\frac34M2^{-j}$, hence $M2^{-j-1}\le\frac23\mu$ and
$\mu\ge\frac34M2^{-(K-1)}=\frac32\cdot2^{K-2}$.  The Chernoff bound
$\Pr[X\le(1-\gamma)\mu]\le e^{-\gamma^2\mu/2}$ \cite{Chernoff1952} with
$\gamma=1/3$ gives
\[
 \Pr\bigl[|C(W)|<M2^{-j-1}\bigr]\le e^{-\mu/18}\le e^{-2^{K-2}/12}.
\]
A union bound over the at most $\sum_{j=1}^{K-1}M^j\le2M^{K-1}$ sets $W$
shows that $G$ fails to be expanding with probability at most
$\pi_K=2M^{K-1}e^{-2^{K-2}/12}$.  For the numerical claim, $M=4^{K-1}$
gives
\[
 \ln\pi_K=\bigl(1+2(K-1)^2\bigr)\ln2-\frac{2^{K-2}}{12},
\]
which equals $339\ln2-4096/12<-106$ at $K=14$, and
$\ln\pi_{K+1}-\ln\pi_K=2(2K-1)\ln2-2^{K-2}/12<0$ for all $K\ge14$, since
the subtracted term doubles at each step while the first term grows by
$4\ln2$.  Hence $\pi_K<e^{-100}$ for all $K\ge14$.

\emph{The algorithm.}  Set $W_0=\emptyset$.  For $j=0,\ldots,K-1$, given
the clique $W_j=\{v_1,\ldots,v_j\}$, run the capped search of
\cref{lem:capped-search} on $[M]$ for the predicate ``$u\notin W_j$ and
$G(v_i,u)=1$ for all $i\le j$'', with density parameter
$\lambda_j=2^{-j-1}$, repeating the block $O(\log(K/\eta))$ times so that
its failure probability is at most $\eta/K$ whenever the density promise
holds.  The predicate costs $O(j+1)$ edge queries, including uncomputation.
Let $v_{j+1}$ be the verified output and $W_{j+1}=W_j\cup\{v_{j+1}\}$.  If
any search fails, output $\bot$; otherwise output $W_K$.

\emph{Correctness and cost.}  Every accepted vertex was verified to be
adjacent to all earlier vertices, so a non-$\bot$ output is a $K$-clique
for every input graph.  If $G$ is expanding, then the search at level $j$
has marked density at least $\lambda_j$, and a union bound over the $K$
searches shows that the algorithm outputs $W_K$ with probability at least
$1-\eta$.  Hence it succeeds with probability at least $1-\eta-\pi_K$ over
$G$ and its internal randomness.  Each capped block at level $j$ uses
$O(2^{j/2})$ predicate evaluations, so the total number of edge queries is
\[
 O\!\left(\log\frac K\eta\sum_{j=0}^{K-1}(j+1)2^{j/2}\right)
 =O\!\left(K2^{K/2}\log\frac K\eta\right)
\]
on every input, because every block has a predetermined cap.
\end{proof}

\begin{proof}[Proof of \cref{cor:separation}]
The subgraph induced on the first $M=4^{K-1}\le N$ vertices is distributed
as $G(M,1/2)$, so the first part follows from \cref{prop:random-graph},
using $2^{K/2}\le2^{n/4+1/2}$ and $K=O(\log N)$.  The second part is the
distributional statement of \cref{prop:classical-lower}, which applies
because $K\ge8$.  For the final sentence, note that $n\ge26$ gives
$K\ge14$, so $\pi_K<e^{-100}$; with $\eta=1/4$ the quantum algorithm
succeeds with probability greater than $2/3$, while $5/8<2/3$.
\end{proof}

\subsection{Proof of the quantum lower bound}
\label{app:quantum-lower}

\begin{proof}[Proof of \cref{prop:quantum-lower}]
Let $m=\lfloor n/4\rfloor\ge2$ and $H=2^m$.  By Erd\H{o}s's bound
$R(k,k)>2^{k/2}$ for $k\ge3$ \cite{Erdos1947}, there is a graph $A_0$ on
$[H]$ with no clique or independent set of order $2m$; fix one and hardwire
it, so that evaluating $A_0$ costs no queries.  For $h\colon[N]\to[H]$ let
$A_h$ be the graph \cref{eq:product-graph}.  It is a simple graph, and one
coherent query to $A_h$ costs $O(1)$ queries to $h$.

Every homogeneous set $S$ of $A_h$ with $|S|\ge2m$ contains two vertices
with the same $h$-value.  Indeed, if $S$ is independent, then its $h$-values
are distinct, because equal values force an edge, and $h(S)$ is an
independent set of $A_0$; hence $|S|=|h(S)|\le2m-1$, which is impossible.
So $S$ is a clique, $h(S)$ is a clique of $A_0$, and $|h(S)|\le2m-1<|S|$.

Since $2m\le\lfloor n/2\rfloor$, a set of the order in the statement
qualifies.  Given a $q$-query algorithm $\mathcal A$ as in the statement,
run it on $A_h$, evaluate $h$ on the $\lfloor n/2\rfloor$ returned
vertices, and output two of them with equal $h$-value.  This finds a
collision of $h$ with probability at least $2/3$ using $O(q+n)$ queries to
$h$, for every $h$.  Finding a collision in a uniformly random function
$[N]\to[H]$ with $N\ge H$ requires $\Omega(H^{1/3})$ quantum queries
\cite{Zhandry2015,LiuZhandry2019}, so $q+O(n)=\Omega(2^{m/3})$.  Finally
$m\ge(n-3)/4$ gives $2^{m/3}\ge2^{-1/4}N^{1/12}$, and $n=o(N^{1/12})$, so
$q=\Omega(N^{1/12})$.
\end{proof}

%% file: appendices/D_extensions_and_reproducibility.tex
\section{Multicolour extension, numerical illustrations, and AI disclosure}
\label{app:extensions}

\subsection{A multicolour extension of size-biased survival}
\label{app:multicolour}

Suppose every edge has one of $q\ge2$ colours, and define
\[
 r=q(K-2)+1,
 \qquad
 B_0=0,
 \qquad
 B_d=1+q+\cdots+q^{d-1}=\frac{q^d-1}{q-1}\quad(d\ge1).
\]
Choose the power of two
\[
 L_q=2^{\lceil\log_2(2B_r)\rceil},
\]
so $2B_r\le L_q<4B_r$.

\begin{corollary}[Multicolour size-biased Ramsey search]
\label{cor:multicolour}
For $q,K\ge2$, every $q$-edge-coloured complete graph on at least $L_q$
vertices admits a bounded-error quantum algorithm that outputs a
monochromatic $K$-clique using
\[
 O\!\left(r^2\sqrt{L_q}\log(2r)\log\frac1\eta\right)
\]
colour-oracle queries in the worst case, where $r$, $B_r$, and $L_q$ are as
defined above.
\end{corollary}

\begin{proof}
Run the size-biased recursion of \cref{sec:algorithm} for $r$ rounds on the
first $L_q$ vertices, now following the colour of a uniformly sampled edge
among $q$ possibilities.  Let $P_d^{(q)}(s)$ be the minimum probability that
the ideal recursion completes $d$ further rounds from a candidate set of size
$s$ and a fixed pivot.  We claim that
\[
 P_d^{(q)}(s)\ge \frac{(s-B_d)_+}{s}
\]
for all $d\ge0$ and $s\ge1$.  The case $d=0$ is immediate.  For $d\ge1$ and
$s\ge2$, let the $q$ colour classes among the $s-1$ non-pivot vertices have
sizes $a_1,\ldots,a_q$.  The induction hypothesis gives
\[
 P_d^{(q)}(s)
 \ge \frac{\sum_{j=1}^q a_jP_{d-1}^{(q)}(a_j)}{s-1}
 \ge \frac{\sum_{j=1}^q(a_j-B_{d-1})_+}{s-1}.
\]
Since $\sum_j a_j=s-1$ and $B_d=1+qB_{d-1}$, the numerator is at least
$(s-B_d)_+$.  This proves the claimed bound after weakening the denominator
from $s-1$ to $s$.  The preceding display also gives the stronger bound
$(s-B_d)_+/(s-1)$ whenever $d\ge1$ and $s\ge2$.  Consequently, starting from
$L_q\ge2B_r$, the ideal recursion survives all $r$ rounds with probability
strictly greater than $1/2$.

A membership predicate at level $i$ contains $i$ colour constraints and costs
$O(i+1)$ oracle queries.  Using the universal density bound $1/L_q$, one
capped sample costs $O((i+1)\sqrt{L_q})$ colour queries.  Repeating the capped
block $O(\log(2r))$ times makes each sampling failure probability
$O(1/r)$, so one fixed-cap run costs
\[
 O\!\left(
 \sqrt{L_q}\log(2r)\sum_{i=0}^{r-1}(i+1)
 \right)
 =O\!\left(r^2\sqrt{L_q}\log(2r)\right).
\]
The union bound preserves a positive constant success probability.  Since
$r=q(K-2)+1$, one of the $q$ recorded colours occurs at least
$\lceil r/q\rceil=K-1$ times.  The corresponding pivots and the final vertex
form a monochromatic $K$-clique by the same nesting argument as in
\cref{lem:output}.  Independent repetition and direct verification amplify
the success probability to $1-\eta$ at an additional
$O(\log(1/\eta))$ factor.
\end{proof}

\subsection{Small-instance numerical illustrations}
\label{app:experiments}

\subparagraph*{Ideal-sampler recursion.}
We ran the scale-aware recursion at $K=3$, $N=16$, and $r=3$ for $1{,}000$
trials in each of five graph families.  The error schedule was
$(\varepsilon_0,\varepsilon_1,\varepsilon_2)=(1/64,1/64,1/32)$, with
sample counts
$m_i=\lceil\log(600)/(2\varepsilon_i^2)\rceil$, corresponding to a total
estimation-failure budget of $0.01$.  Pivots were sampled uniformly,
majority-estimation samples were independent with replacement, and ties
selected colour~$1$.  We replaced quantum search by ideal uniform sampling to
isolate the combinatorial recursion.

The deterministic inputs were the complete graph, the empty graph, and the
balanced complete bipartite graph $K_{8,8}$.  For the random inputs, one
$G(16,1/2)$ graph was held fixed across all $1{,}000$ trials, while the final
family drew a fresh $G(16,1/2)$ graph in each trial.  Every returned triple was
verified; exhaustion or a nonhomogeneous output counted as failure.

\begin{center}
\begin{tabular}{lrr}
\hline
Graph family on 16 vertices & failed/trials & minimum final $|S_3|$\\
\hline
complete & $0/1000$ & 13\\
empty & $0/1000$ & 13\\
balanced complete bipartite $K_{8,8}$ & $0/1000$ & 6\\
one fixed $G(16,1/2)$ & $0/1000$ & 2\\
fresh $G(16,1/2)$ per trial & $0/1000$ & 2\\
\hline
\end{tabular}
\end{center}

All $5{,}000$ trials returned homogeneous triples.  The final column records
the smallest observed value of $|S_3|$ in each family.

\subparagraph*{Isolated Grover sampler.}
We evolved a $16$-dimensional state vector initialized in the uniform
superposition, using marked-set sizes $1,2,3,4,8,15$.  The measurement
distribution was averaged over a uniformly chosen number of Grover iterations
from zero through four.  The probability of measuring a marked item ranged
from $0.4026746749$ to $0.5973253251$.  Conditional on a marked outcome, the
largest deviation from the uniform distribution on marked items was zero to
machine precision.  This calculation isolates the permutation symmetry used
in \cref{lem:capped-search}.

\subparagraph*{Finite split-tree survival.}
Minimizing the ideal size-biased survival probability over every admissible
colour-class split gives exact finite-depth comparisons with the analytic
bounds.  For nine binary levels starting from $1{,}024$ candidates, the
minimum was approximately $0.5086135566$, above the lower bound $513/1024$.
For four ternary levels starting from $128$ candidates, the minimum was
approximately $0.7155732457$, above $88/127$.

\subsection{AI Disclosure}
\label{app:ai-disclosure}

We used ChatGPT and Claude as auxiliary tools for literature search,
brainstorming, proof checking, drafting and editing, and building experiment
scripts (Sections~1, 3--6, and Appendices~A--D).  The authors
determined the research direction and presentations, evaluated and revised
all AI-assisted material, and take full responsibility for the correctness,
originality, and final content of the paper, including every proof and
reference.

%% file: main.bbl
\begin{thebibliography}{10}

\bibitem{BianEtAl2013}
Zhengbing Bian, Fabian Chudak, William~G. Macready, Geordie Rose, and Frank
  Gaitan.
\newblock Experimental determination of {R}amsey numbers.
\newblock {\em Physical Review Letters}, 111:130505, 2013.
\newblock URL: \url{https://arxiv.org/abs/1201.1842}, \href
  {http://arxiv.org/abs/1201.1842} {\path{arXiv:1201.1842}}, \href
  {https://doi.org/10.1103/PhysRevLett.111.130505}
  {\path{doi:10.1103/PhysRevLett.111.130505}}.

\bibitem{BBHT1998}
Michel Boyer, Gilles Brassard, Peter H{\o}yer, and Alain Tapp.
\newblock Tight bounds on quantum searching.
\newblock {\em Fortschritte der Physik}, 46(4--5):493--506, 1998.
\newblock URL: \url{https://arxiv.org/abs/quant-ph/9605034}, \href
  {http://arxiv.org/abs/quant-ph/9605034} {\path{arXiv:quant-ph/9605034}}.

\bibitem{Chernoff1952}
Herman Chernoff.
\newblock A measure of asymptotic efficiency for tests of a hypothesis based on
  the sum of observations.
\newblock {\em The Annals of Mathematical Statistics}, 23(4):493--507, 1952.
\newblock \href {https://doi.org/10.1214/aoms/1177729330}
  {\path{doi:10.1214/aoms/1177729330}}.

\bibitem{ChildsEisenberg2005}
Andrew~M. Childs and Jason~M. Eisenberg.
\newblock Quantum algorithms for subset finding.
\newblock {\em Quantum Information and Computation}, 5(7):593--604, 2005.
\newblock \href {http://arxiv.org/abs/quant-ph/0311038}
  {\path{arXiv:quant-ph/0311038}}, \href {https://doi.org/10.26421/QIC5.7-7}
  {\path{doi:10.26421/QIC5.7-7}}.

\bibitem{ConlonFoxGrinshpunHe2019}
David Conlon, Jacob Fox, Andrey Grinshpun, and Xiaoyu He.
\newblock Online {R}amsey numbers and the subgraph query problem.
\newblock In {\em Building Bridges II: Mathematics of L{\'a}szl{\'o}
  Lov{\'a}sz}, volume~28 of {\em Bolyai Society Mathematical Studies}, pages
  159--194. Springer, 2019.
\newblock URL: \url{https://arxiv.org/abs/1806.09726}, \href
  {http://arxiv.org/abs/1806.09726} {\path{arXiv:1806.09726}}, \href
  {https://doi.org/10.1007/978-3-662-59204-5_4}
  {\path{doi:10.1007/978-3-662-59204-5_4}}.

\bibitem{Doern2005}
Sebastian D{\"o}rn.
\newblock Quantum complexity bounds for independent set problems.
\newblock {\em arXiv preprint}, 2005.
\newblock URL: \url{https://arxiv.org/abs/quant-ph/0510084}, \href
  {http://arxiv.org/abs/quant-ph/0510084} {\path{arXiv:quant-ph/0510084}}.

\bibitem{Erdos1947}
Paul Erd\H{o}s.
\newblock Some remarks on the theory of graphs.
\newblock {\em Bulletin of the American Mathematical Society}, 53(4):292--294,
  1947.
\newblock \href {https://doi.org/10.1090/S0002-9904-1947-08785-1}
  {\path{doi:10.1090/S0002-9904-1947-08785-1}}.

\bibitem{ErdosSzekeres1935}
Paul Erd\H{o}s and George Szekeres.
\newblock A combinatorial problem in geometry.
\newblock {\em Compositio Mathematica}, 2:463--470, 1935.
\newblock URL: \url{https://eudml.org/doc/88611}.

\bibitem{GaitanClark2012}
Frank Gaitan and Lane Clark.
\newblock {R}amsey numbers and adiabatic quantum computing.
\newblock {\em Physical Review Letters}, 108:010501, 2012.
\newblock URL: \url{https://arxiv.org/abs/1103.1345}, \href
  {http://arxiv.org/abs/1103.1345} {\path{arXiv:1103.1345}}.

\bibitem{LeGallNishimuraTani2016}
Fran{\c c}ois~Le Gall, Harumichi Nishimura, and Seiichiro Tani.
\newblock Quantum algorithms for finding constant-sized sub-hypergraphs.
\newblock {\em Theoretical Computer Science}, 609:569--582, 2016.
\newblock \href {http://arxiv.org/abs/1310.4127} {\path{arXiv:1310.4127}},
  \href {https://doi.org/10.1016/j.tcs.2015.10.006}
  {\path{doi:10.1016/j.tcs.2015.10.006}}.

\bibitem{Hoeffding1963}
Wassily Hoeffding.
\newblock Probability inequalities for sums of bounded random variables.
\newblock {\em Journal of the American Statistical Association},
  58(301):13--30, 1963.
\newblock \href {https://doi.org/10.1080/01621459.1963.10500830}
  {\path{doi:10.1080/01621459.1963.10500830}}.

\bibitem{ImpagliazzoNaor1988}
Russell Impagliazzo and Moni Naor.
\newblock Decision trees and downward closures.
\newblock In {\em Proceedings of the Third Annual Structure in Complexity
  Theory Conference}, pages 29--38. IEEE, 1988.
\newblock \href {https://doi.org/10.1109/SCT.1988.5260}
  {\path{doi:10.1109/SCT.1988.5260}}.

\bibitem{JainLiRobereXun2024}
Siddhartha Jain, Jiawei Li, Robert Robere, and Zhiyang Xun.
\newblock On pigeonhole principles and {R}amsey in {TFNP}.
\newblock In {\em 2024 IEEE 65th Annual Symposium on Foundations of Computer
  Science (FOCS)}, pages 406--428. IEEE, 2024.
\newblock \href {http://arxiv.org/abs/2401.12604} {\path{arXiv:2401.12604}},
  \href {https://doi.org/10.1109/FOCS61266.2024.00033}
  {\path{doi:10.1109/FOCS61266.2024.00033}}.

\bibitem{KomargodskiNaorYogev2019}
Ilan Komargodski, Moni Naor, and Eylon Yogev.
\newblock White-box vs. black-box complexity of search problems: {R}amsey and
  graph property testing.
\newblock {\em Journal of the ACM}, 66(5):34:1--34:28, 2019.
\newblock \href {https://doi.org/10.1145/3341106} {\path{doi:10.1145/3341106}}.

\bibitem{Krajicek2001}
Jan Kraj{\'\i}\v{c}ek.
\newblock On the weak pigeonhole principle.
\newblock {\em Fundamenta Mathematicae}, 170(1--2):123--140, 2001.
\newblock \href {https://doi.org/10.4064/fm170-1-8}
  {\path{doi:10.4064/fm170-1-8}}.

\bibitem{Krajicek2005}
Jan Kraj{\'\i}\v{c}ek.
\newblock Structured pigeonhole principle, search problems and hard
  tautologies.
\newblock {\em The Journal of Symbolic Logic}, 70(2):619--630, 2005.
\newblock \href {https://doi.org/10.2178/jsl/1120224731}
  {\path{doi:10.2178/jsl/1120224731}}.

\bibitem{LeeMagniezSantha2012}
Troy Lee, Fr{\'e}d{\'e}ric Magniez, and Miklos Santha.
\newblock A learning graph based quantum query algorithm for finding
  constant-size subgraphs.
\newblock {\em Chicago Journal of Theoretical Computer Science}, 2012(10),
  2012.
\newblock URL: \url{https://arxiv.org/abs/1109.5135}, \href
  {http://arxiv.org/abs/1109.5135} {\path{arXiv:1109.5135}}.

\bibitem{LiuZhandry2019}
Qipeng Liu and Mark Zhandry.
\newblock On finding quantum multi-collisions.
\newblock In {\em Advances in Cryptology---EUROCRYPT 2019, Part III}, volume
  11478 of {\em Lecture Notes in Computer Science}, pages 189--218. Springer,
  2019.
\newblock \href {http://arxiv.org/abs/1811.05385} {\path{arXiv:1811.05385}},
  \href {https://doi.org/10.1007/978-3-030-17659-4_7}
  {\path{doi:10.1007/978-3-030-17659-4_7}}.

\bibitem{Montanaro2018}
Ashley Montanaro.
\newblock Quantum walk speedup of backtracking algorithms.
\newblock {\em Theory of Computing}, 14(15):1--24, 2018.
\newblock URL: \url{https://arxiv.org/abs/1509.02374}, \href
  {http://arxiv.org/abs/1509.02374} {\path{arXiv:1509.02374}}, \href
  {https://doi.org/10.4086/toc.2018.v014a015}
  {\path{doi:10.4086/toc.2018.v014a015}}.

\bibitem{PasarkarPapadimitriouYannakakis2023}
Amol Pasarkar, Christos Papadimitriou, and Mihalis Yannakakis.
\newblock Extremal combinatorics, iterated pigeonhole arguments and
  generalizations of {PPP}.
\newblock In {\em 14th Innovations in Theoretical Computer Science Conference
  (ITCS 2023)}, volume 251 of {\em Leibniz International Proceedings in
  Informatics (LIPIcs)}, pages 88:1--88:20. Schloss Dagstuhl--Leibniz-Zentrum
  f{\"u}r Informatik, 2023.
\newblock \href {https://doi.org/10.4230/LIPIcs.ITCS.2023.88}
  {\path{doi:10.4230/LIPIcs.ITCS.2023.88}}.

\bibitem{PionMniszewski2025}
Joel~E. Pion and Susan~M. Mniszewski.
\newblock Toward computing bounds for {R}amsey numbers using quantum annealing.
\newblock {\em Quantum Information Processing}, 24(7):221, 2025.
\newblock URL: \url{https://arxiv.org/abs/2311.04405}, \href
  {http://arxiv.org/abs/2311.04405} {\path{arXiv:2311.04405}}, \href
  {https://doi.org/10.1007/s11128-025-04839-x}
  {\path{doi:10.1007/s11128-025-04839-x}}.

\bibitem{QuLiWangBaoCao2013}
Ri~Qu, Zong shang Li, Juan Wang, Yan ru~Bao, and Xiao chun Cao.
\newblock Computing hypergraph {R}amsey numbers by using quantum circuit.
\newblock {\em Quantum Information Processing}, 12(7):2487--2496, 2013.
\newblock \href {http://arxiv.org/abs/1210.3419} {\path{arXiv:1210.3419}},
  \href {https://doi.org/10.1007/s11128-013-0541-9}
  {\path{doi:10.1007/s11128-013-0541-9}}.

\bibitem{RanjbarMacreadyClarkGaitan2016}
Mani Ranjbar, William~G. Macready, Lane Clark, and Frank Gaitan.
\newblock Generalized {R}amsey numbers through adiabatic quantum optimization.
\newblock {\em Quantum Information Processing}, 15(9):3519--3542, 2016.
\newblock \href {http://arxiv.org/abs/1606.01078} {\path{arXiv:1606.01078}},
  \href {https://doi.org/10.1007/s11128-016-1363-3}
  {\path{doi:10.1007/s11128-016-1363-3}}.

\bibitem{Wang2016}
Hefeng Wang.
\newblock Determining {R}amsey numbers on a quantum computer.
\newblock {\em Physical Review A}, 93(3):032301, 2016.
\newblock \href {http://arxiv.org/abs/1510.01884} {\path{arXiv:1510.01884}},
  \href {https://doi.org/10.1103/PhysRevA.93.032301}
  {\path{doi:10.1103/PhysRevA.93.032301}}.

\bibitem{Zhandry2015}
Mark Zhandry.
\newblock A note on the quantum collision and set equality problems.
\newblock {\em Quantum Information and Computation}, 15(7--8):557--567, 2015.
\newblock \href {http://arxiv.org/abs/1312.1027} {\path{arXiv:1312.1027}},
  \href {https://doi.org/10.26421/QIC15.7-8-2}
  {\path{doi:10.26421/QIC15.7-8-2}}.

\bibitem{Zhu2012}
Yechao Zhu.
\newblock Quantum query complexity of constant-sized subgraph containment.
\newblock {\em International Journal of Quantum Information}, 10(3):1250019,
  2012.
\newblock \href {http://arxiv.org/abs/1109.4165} {\path{arXiv:1109.4165}},
  \href {https://doi.org/10.1142/S0219749912500190}
  {\path{doi:10.1142/S0219749912500190}}.

\end{thebibliography}
